\documentclass{llncs}

\usepackage{amsmath,amssymb}
\usepackage{enumitem}
\usepackage[dvips]{epsfig}
\usepackage{epsf}
\usepackage{float}
\usepackage{array}

\usepackage{pgf}
\usepackage{tikz}
\usetikzlibrary{arrows,automata}

\def\lss{{\rm lss}}
\def\sc{{\rm sc}}
\def\lcm{{\rm lcm}}

\def\Zee{{\mathbb Z}}

\def\fmod#1 #2{#1\ ({\rm mod}\ #2)}

\def \endpf{{\ \ $\Box$ \medbreak}}

\newcommand{\ar}{\rightarrow}

\begin{document}

\frontmatter

\title{Decidability and Shortest Strings in Formal Languages}

\author{Levent Alpoge\inst{1}, Thomas Ang\inst{2}, Luke Schaeffer\inst{2},
and Jeffrey Shallit\inst{2}}

\institute{Harvard College,
Cambridge, MA  02138, USA \\
\email{levent.alpoge@gmail.com}
\and
School of Computer Science,
University of Waterloo,
Waterloo, ON  N2L 3G1 Canada  \\
\email{angthomas@gmail.com, l3schaef@uwaterloo.ca, shallit@cs.uwaterloo.ca}
}

\maketitle 

\begin{abstract}
Given a formal language $L$ specified in various ways, we consider
the problem of determining if $L$ is nonempty.  If $L$ is indeed
nonempty, we find upper and lower bounds on the length of the
shortest string in $L$.
\end{abstract}

\section{Introduction}

Given a formal language $L$ specified in some finite way, a common problem
is to determine whether $L$ is nonempty.  And if $L$ is indeed nonempty,
then another common problem is to determine good upper and lower bounds
on the length of the shortest string in $L$, which we write
as $\lss(L)$.
Such bounds can be useful,
for example, in estimating the state complexity of $L$, since
$\lss(L) < \sc(L)$.

As an example, we start with a very simple
result often stated in introductory classes on formal language theory.

\begin{proposition}
Let $L$ be accepted by an NFA $M$ with $n$ states and $t$
transitions.
Then we can decide in time $O(n + t)$ whether
$L \not= \emptyset$.  If $L$ is nonempty, then
$\lss(L) < n$.  Further, this bound is tight.
\label{one}
\end{proposition}

%
%

We now turn to a more challenging example.  Here $L$ is specified as the
{\it complement} of a language accepted by an NFA.

\begin{theorem} 
Let $L$ be accepted by an NFA with $n$ states.  Then it is 
PSPACE-complete to determine whether $\overline{L} \not= \emptyset$.
If $\overline{L} \not= \emptyset$, then 
$\lss( \, \overline{L} \, ) < 2^n$.
Further, for some constant $c$, $0 < c \leq 1$,
there is an infinite family of examples with
$n$ states such that $\lss( \, \overline{L} \, ) \geq 2^{cn}$.
\end{theorem}

\begin{proof}
For the PSPACE-completeness, see \cite{Aho&Hopcroft&Ullman:1974}.

The upper bound is easy and
follows from the subset construction.  The lower
bound is significantly harder; see \cite{Ellul:2005}. \endpf
\end{proof}

These two examples set the theme of the paper.  We examine several problems
about shortest strings in regular languages and prove bounds for
$\lss(L)$.  Some of the results have appeared in the master's thesis of the
second author \cite{Ang:2010}.

\section{The first problem}

Recall the following classical result about intersections of regular
languages.

\begin{proposition}
Let $L_1$ (resp., $L_2$) be accepted by an NFA
with $s_1$ states and $t_1$ transitions 
(resp., $s_2$ states and $t_2$ transitions)
Then $L_1 \ \cap \ L_2$ is accepted by
an NFA with $s_1 s_2$ states and
$t_1 t_2$ transitions.
\label{two}
\end{proposition}

\begin{proof}
Use the usual direct product construction. \endpf
\end{proof}

This suggests the following natural problems.  Given NFA's $M_1$ and
$M_2$ as above, decide if $L(M_1) \ \cap \ L(M_2) \not= \emptyset$.
This can clearly be done in $O(s_1 s_2 + t_1 t_2)$ time, by using the
direct product construction followed by breadth-first or depth-first
search.

Now assume
$L(M_1) \ \cap \ L(M_2) \not= \emptyset$. What is a good bound
on $\lss(L(M_1) \ \cap \ L(M_2))$?  Combining Propositions~\ref{one} and
\ref{two}, we immediately
get the upper bound $\lss(L(M_1) \ \cap \ L(M_2)) < s_1 s_2$.

However, is this bound tight?  For $\gcd(s_1, s_2) = 1$
an obvious construction shows it is, even in the
unary case:  choose $L_1 = a^{s_1-1} (a^{s_1})^*$ and
$L_2 = a^{s_2-1} (a^{s_2})^*$.  
However, this idea no longer works for $\gcd(s_1,s_2) > 1$.  Nevertheless,
the bound $s_1 s_2 -1$ is tight for binary and larger alphabets, as the
following result shows.

\begin{theorem}
\label{thm:intersect}
For all integers $m, n \geq 1$ there exist DFAs $M_1, M_2$ with $m$ and
$n$ states, respectively, and with $|\Sigma| = 2$ such that $L(M_1)
\cap L(M_2) \not= \emptyset$, and ${\rm lss}(L(M_1) \cap L(M_2) ) =
mn-1$.
\end{theorem}

\begin{proof}
The proof is constructive. Without loss of generality, assume $m \leq
n$, and set $\Sigma = \{0,1\}$. Let $M_1$ be the DFA given
by $(Q_1, \Sigma,
\delta_1, p_0, F_1)$, where $Q_1 = \{p_0, p_1, p_2,\ldots, p_{m-1}\}$,
$F_1 = p_0$, and for each $a$, $0 \leq a \leq m-1$, and 
$c \in \lbrace 0, 1 \rbrace $ we set
\begin{equation}
\label{eq:delta1}
\delta_1(p_a, c) = p_{(a+c) \bmod m}.
\nonumber
\end{equation}
Then $$L(M_1) = \{ x \in \Sigma^* : |x|_1 \equiv 0\!\! \pmod {m} \}.$$ 

Let $M_2$ be the DFA $(Q_2, \Sigma, \delta_2, q_0, F_2)$, shown in Figure \ref{fig:m2}, where $Q_2 = \{q_0, q_1,\ldots, q_{n-1}\}$, $F_2 = q_{n-1}$, and for each $a$, $0 \leq a \leq n-1$,
\begin{equation}
\label{eq:delta2}
\delta_2(q_a, c) =
\begin{cases}
q_{a+c}, & \text{if } 0 \leq a < m-1; \\
q_{(a+1) \bmod n}, & \text{if }c = 0 \text{ and } m-1 \leq a \leq n-1; \\
q_0, & \text{if }c = 1 \text{ and } m-1 \leq a \leq n-1.
\end{cases} \nonumber
\end{equation}

\begin{figure}[H]
\begin{center}
\resizebox{\columnwidth}{!}{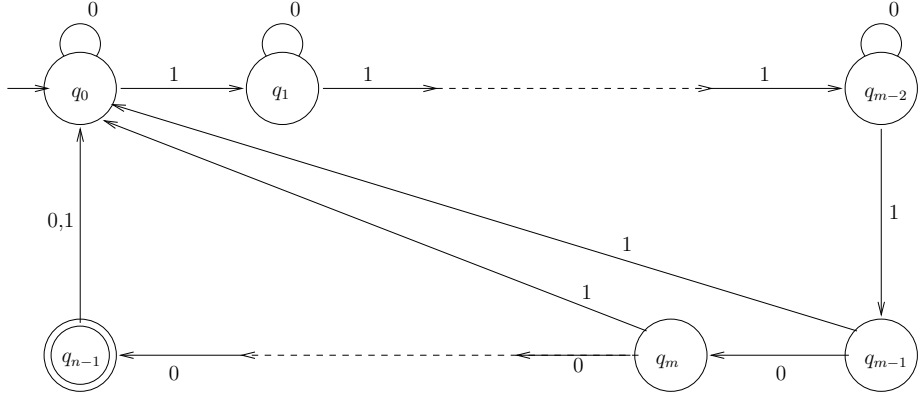}
\end{center}

\caption{
	The DFA $M_2$. 
}
\label{fig:m2}
\end{figure}

Focusing solely on the $1$'s that appear in some accepting computation in
$M_2$, we see that we can return to $q_0$ 
\begin{itemize}
\item[(a)] via a simple path with $m$ $1$'s, or
\item[(b)] (if we go through $q_{n-1}$), via a simple path with 
$(m-1)$ $1$'s and ending in the transition $\delta(q_{n-1}, 0) = q_0$.
\end{itemize}
After some number of cycles through $q_0$, we eventually
arrive at $q_{n-1}$.  Letting $i$ denote the number of times a path of
type (b) is chosen (including the last path that arrives at $q_{n-1}$)
and $j$ denote the number of times a path of type (a) is chosen, we see
that the number of $1$'s in any accepted word must be of the form
$i(m-1) + jm$, with $i > 0$, $j \geq 0$.  The number of $0$'s along
such a path is then at least $i(n-m+1) - 1$, with the $-1$ in this expression
arising from the fact that the last part of the path terminates at $q_{n-1}$
without taking an additional $0$ transition back to $q_0$.

Thus
\begin{align*}
L(M_2) \subseteq \{ x \in \Sigma^* : \exists i,j \in {\mathbb N}, \text{ such that } i > 0, j \geq 0, \text{ and } \\
|x|_1  = i(m-1) + jm, \ |x|_0 \geq i(n-m+1)-1 \}.
\end{align*}
Furthermore,
for every $i,j \in {\mathbb N}, \text{ such that } i > 0, j \geq 0$,
there exists an $x \in L(M_2)$ such that $|x|_1  = i(m-1) + jm$, and
$|x|_0 = i(n-m+1)-1$.   This is obtained, for example, by cycling 
$j$ times from $q_0$ to $q_{m-1}$
and then back to $q_0$ via a transition on $1$, then $i-1$ times
from $q_0$ to $q_{n-1}$ and then back to $q_0$ via a transition on $0$,
and finally one more time from $q_0$ to $q_{n-1}$.  

It follows then that 
\begin{align*}
L(M_1) \cap L(M_2) &\subseteq \{ x \in \Sigma^*: \exists i,j \in {\mathbb N}, \text{ such that } i > 0, j \geq 0, \text{ and } \\
& |x|_1  = i(m-1) + jm, \  |x|_0 \geq i(n-m+1)-1\\
& \text{and } i(m-1) + jm \equiv 0\!\! \pmod {m} \}.
\end{align*}
Further, for every such $i$ and $j$, there exists a corresponding element in $L(M_1 \cap M_2)$. Since $m-1$ and $m$ are relatively prime, the shortest such word corresponds to $i = m$, $j = 0$,  and satisfies $|x|_0 = m(n-m+1)-1$. In particular, a shortest accepted word is $(1^{m-1}0^{n-m+1})^{m-1}1^{m-1}0^{n-m}$, which is of length $mn-1$. \endpf
\end{proof}

We can also obtain a bound for the unary case.
Let 
$$F(m,n) = \max_{{1 \leq i \leq m}\atop{1\leq j \leq n}}
\bigl( \, \max(m-i, n-j)+ \lcm(i,j) \, \bigr),$$
as defined in \cite{Pighizzini&Shallit:2002}.

\begin{theorem}
Given unary DFA's $M_1$ (resp., $M_2$) with $m$ (resp., $n$) states, 
accepting $L_1$ (resp., $L_2$), we have $\lss(L_1 \ \cap \ L_2) \leq
F(m,n)-1$.  Furthermore, for all $m, n \geq 1$ there exist
unary DFA's of $m$ and $n$ states achieving this bound.
\end{theorem}

\begin{proof}
Follows from \cite{Pighizzini&Shallit:2002}.
\end{proof}

\section{The second problem}

Recall the Post correspondence problem:  we are given two finite
nonempty languages $A = \lbrace x_1, x_2, \ldots, x_n \rbrace$ and $B = 
\lbrace y_1, y_2, \ldots, y_n \rbrace$, and we want to determine if there
exist $r \geq 1$ and
a finite sequence of indices $i_1, i_2, \ldots , i_r$
such that $x_{i_1} \cdots x_{i_r} = y_{i_1} \ldots y_{i_r}$.  As is well-known,
this problem is undecidable.

Levent Alpoge \cite{Alpoge} asked about the variant where
we throw away the ``correspondence'':
determine if there exist $r, s \geq 1$ and two finite sequences of
indices $i_1, \ldots, i_r$ and $j_1, \ldots, j_s$ such that
$x_{i_1} \cdots x_{i_r} = y_{j_1} \cdots y_{j_s}$.   In other words,
we want to decide if $A^+ \ \cap \ B^+ \not= \emptyset$.

This variant is, of course, decidable.   In fact, even a more general
version is decidable, where the languages need not be finite.

\begin{proposition}
Suppose $A$ is a language accepted by an NFA $M_1$
with $s_1$ states and $t_1$ transitions, and $B$ is accepted by
an NFA $M_2$ with $s_2$ states and $t_2$ transitions.  
Then we can decide in $O(s_1 s_2 + t_1 t_2)$ time whether
$A^+ \ \cap \ B^+ \not= \emptyset$.
\end{proposition}

\begin{proof}
Given NFA $M_1 = (Q_1, \Sigma, \delta_1, q_1, F_1)$
accepting $A$,
we can create an NFA-$\epsilon$
$M'_1 = (Q_1, \Sigma, \delta'_1, q_1, F'_1)$ accepting $A^+$
by adding an $\epsilon$-transition from every final state of $M_1$
back to $q_0$.
We can apply a similar construction to 
create $M'_2 = (Q_2, \Sigma, \delta'_2,
q_2, F'_2)$ accepting
$B^+$.  Then we can create an NFA-$\epsilon$ $M$
accepting $A^+ \ \cap \ B^+$ using
the usual direct product construction.  Since this construction is
crucial to what follows, and since there is one subtle point,
we describe it in some detail.

Given $M'_1 = (Q_1, \Sigma, \delta'_1, q_1, F'_1)$ and
$M'_2 = (Q_2, \Sigma, \delta'_2, q_2, F'_2)$ as above, 
$M = (Q, \Sigma, \delta, q_0, F)$, where
$Q = Q_1 \times Q_2$, $q_0 = [q_1, q_2]$, and $F = F_1 \times F_2$.
The transition function $\delta$ is defined as follows:  

For $p \in Q_1$, $q \in Q_2$, and $a \in \Sigma \cup \lbrace \epsilon \rbrace$
we have $[p', q'] \in \delta([p,q], a)$ if $p' \in \delta'_1(p,a)$ and
$q' \in \delta'_2(q,a)$.   These transitions correspond to
the usual direct product edges of the
transition diagram.

However, we also need edges in which one machine performs an explicit
$\epsilon$-transition, and the other machine performs an implicit
$\epsilon$-transition by simply staying in its own state.
This corresponds to including the transitions
$[p', q'] \in \delta([p,q],\epsilon)$ 
if $p' \in \delta'_1(p, \epsilon)$ and $q = q'$
or if $p' = p$ and $q' \in \delta'_2(q, \epsilon)$.

This construction results in an NFA-$\epsilon$ accepting
$A^+ \ \cap \ B^+$ and having at most $t_1 t_2 + 2 s_1 s_2$
transitions.

Now we can use the usual breadth-first or depth-first search to solve
the emptiness problem.
\endpf
\label{ab}
\end{proof}

\begin{corollary}
Given NFA's $M_1$ accepting $L_1$ (resp., $M_2$ accepting $L_2$)
of $m$ (resp., $n$) states, the
shortest string in $L_1^+ \ \cap \ L_2^+$ is of length at most
$mn-1$.

Suppose $m \geq n \geq 1$. Then there exists 
$M_1$ accepting $L_1$ (resp., $M_2$ accepting $L_2$) of $m$
(resp., $n$) states such that the shortest string in $L_1^+ \ \cap \ L_2^+$
is of length $\geq (m-1)n$.  
\end{corollary}

\begin{proof}
The first assertion follows from Proposition~\ref{ab}.

For the second assertion, we can take $M_1$ and $M_2$ as in the
proof of Theorem~\ref{thm:intersect}.    Clearly $L_1 = L_1^+$.
When we apply our construction to $M_2$ to create $L_2^+$,
we add an $\epsilon$-transition from $q_{n-1}$ back to $q_0$.  The
effect is to allow one less $0$ in each cycle through the states.
As in the proof of Theorem~\ref{thm:intersect}, to get the proper number of
$1$'s, we must have
$i = m$, and hence the shortest string in $L_1^+ \ \cap \ L_2^+$
is of length $(m-1)n$.  \endpf
\end{proof}

We can improve the upper bound to $mn-2$ as follows:

\begin{theorem}
\label{thm:positiveclosure}
For any $m$-state DFA $M_1$ and n-state DFA $M_2$ such that $L(M_1)^+ \cap L(M_2)^+ \neq \emptyset$ we have ${\rm lss}(L(M_1)^+ \cap L(M_2)^+) < mn-1$.
\end{theorem}

\begin{proof}

Assume, contrary to what we want to prove, that we have DFAs $M_1$ and
$M_2$ with $m$ and $n$ states, respectively, such that ${\rm
lss}(L(M_1)^+ \cap L(M_2)^+) = mn-1$. Let $M_1$ be the DFA given by
$(Q_1, \Sigma, \delta_1, p_0, F_1)$, where $Q_1 = \{p_0, p_1,
p_2,\ldots, p_{m-1}\}$, and let $M_2$ be the DFA given by $(Q_2,
\Sigma, \delta_2, p_0, F_2)$, where $Q_2 = \{q_0, q_1, q_2,\ldots,
q_{n-1}\}$. Then let $M_1^\prime$ and $M_2^\prime$ be the
$\epsilon$-NFAs obtained by adding $\epsilon$-transitions from the
final states to the start states in $M_1$ and $M_2$, respectively. Let
$M$ be the $\epsilon$-NFA obtained by applying the cross-product
construction to $M_1^\prime$ and $M_2^\prime$. Then $M$ accepts
$L(M_1)^+ \cap L(M_2)^+$.

If $M$ has more than one final state, a shortest accepting path would
only visit one of them, and this immediately gives a contradiction. So,
assume each of $M_1$ and $M_2$ have only one final state; that is $F_1
= \{p_x\in Q_1\}$ and $F_2 = \{q_y\in Q_2\}$. Then $M = (Q_1\times Q_2,
\Sigma, \delta, [p_0,q_0], {[p_x,q_y]}),$ where for all $p_i \in Q_1,
q_j\in Q_2, a\in \Sigma,
\delta([p_i,q_j],a)=[\delta_1(p_i,a),\delta_2(q_j,a)].$ Note that $M$
has $\epsilon$-transitions from $[p_x,q_j]$ to $[p_0,q_j]$ for all
$q_j\in Q_2$ and $[p_i,q_y]$ to $[p_i,q_0]$ for all $p_i\in Q_1$.

Let $w_1$ be a shortest word accepted by $M_1$ and $w_2$ be a shortest
word accepted by $M_2$. Then $\delta([p_0, q_0], w_1) = [p_x, q_i]$ for
some $i$ such that $q_i \in Q_2$, and while carrying out this
computation we never pass through two states $[p_a,q_b]$ and
$[p_c,q_d]$ such that $a=c$. Likewise, $\delta([p_0, q_0], w_2) = [p_j,
q_y]$ for some $j$ such that $p_j \in Q_1$, and while carrying out this
computation we never pass through two states $[p_a,q_b]$ and
$[p_c,q_d]$ such that $b=d$. If both $x=0$ and $y=0$ the shortest
accepted string is $\epsilon$, so without loss of generality, assume
$x\neq 0$. Then $\delta([p_0, q_0], w_1) = [p_x, q_0]$ or else we can
visit $|w_1| + 2$ states with $|w_1|$ symbols by using an
$\epsilon$-transition and we get a contradiction. If $y = 0$, $w_1$ is
the shortest string accepted by $M$ and we have a contradiction. So,
$y\neq 0$ and $\delta([p_0, q_0], w_2) = [p_0, q_y]$. It follows that
reading $w_1$ from the initial state brings us to $[p_x, q_0]$ without
passing through $[p_0, q_y]$, and reading $w_2$ from the initial state
brings us to $[p_0, q_y]$ without passing through $[p_x, q_0]$. So, a
shortest accepting path need only visit one of $[p_x, q_0]$ and $[p_0,
q_y]$, and again we have a contradiction.  \qed
\end{proof}

We do not know an exact bound for this problem.  However, for the unary
case, we can obtain an exact bound based on a function $G$ introduced
in \cite{Pighizzini&Shallit:2002}.  Define $G(m,n) = \max_{{1 \leq i
\leq m} \atop {1 \leq j \leq n}} \lcm(i,j)$, and define
the variant 
$$G'(m,n) = \max_{{1 \leq i \leq m}\atop{{1 \leq j \leq n} 
\atop{(i,j) \not= (m,n)}}} \lcm(i,j).$$
Then $G'(m,n) = \max(G(m-1,n), G(m,n-1))$.    The function $G$ is a very
difficult one to estimate, although deep results in analytic number
theory give some upper and lower bounds \cite{Pighizzini&Shallit:2002}.

\begin{theorem}
If $M_1$ (resp., $M_2$) is a unary NFA with $m$ states (resp., $n$ states)
and $L_1 = L(M_1)$ (resp., $L_2 = L(M_2)$), then 
$\lss(L_1^+ \ \cap \ L_2^+) \leq G'(m,n)$.  Furthermore, for all
$m, n \geq 1$ there exist unary DFA's of $m$ and $n$ states, respectively,
achieving this bound.
\end{theorem}

\begin{proof}
Assume the input alphabet of both $M_1$ and $M_2$
is $\Sigma = \lbrace a \rbrace$.
Let $c_1$ (resp., $c_2$) be the length of the shortest nonempty
string in $L_1$ (resp., $L_2$).   Clearly $c_1 \leq m$ and
$c_2 \leq n$.  Furthermore, if $c_1 = m$, then $L_1 = (a^m)^*$, and
similarly if $c_2 = n$ then $L_2 = (a^n)^*$.  Hence if
$(c_1,c_2) = (m,n)$, then $\epsilon \in L_1^+ \ \cap \ L_2^+$, and
hence $\lss(L^+1 \ \cap \ L_2^+) = 0 \leq G'(m,n)$.  Otherwise
either $c_1 < m$ or $c_2 < n$.  Without loss of generality, assume
$c_2 < n$.  Then $a^{\lcm(c_1,c_2)} \in L_1^+ \ \cap \ L_2^+$, so
$\lss(L_1^+ \ \cap \ L_2^+) \leq \lcm(c_1, c_2) \leq G(m, n-1) \leq G'(m,n)$.

Now suppose we are given $m$ and $n$.  Let $i, j$ be the integers
maximizing $\lcm(i,j)$ over $1 \leq i \leq m$, $1 \leq j \leq n$
with $(i,j) \not= (m,n)$.  If $i < m$, choose $L_1 = (a^i)^+$, which
can be accepted by a DFA with $i+1 \leq m$ states, and choose
$L_2 = (a^j)^*$, which can be accepted by a DFA with $j \leq n$ states.
Otherwise, reverse the roles of $m$ and $n$.  Thus we get DFA's of
$m$ and $n$ states, respectively, achieving $\lss(L_1^+ \ \cap \ L_2^+)
=G'(m,n)$.  \endpf
\end{proof}

\section{The third problem}

Another variation on the Post correspondence problem,
also proposed by Alpoge \cite{Alpoge}, is more interesting.
Here we throw away only {\it part} of the ``correspondence'':  given
$A = \lbrace x_1, x_2, \ldots, x_n \rbrace$ and $B =
\lbrace y_1, y_2, \ldots, y_n \rbrace$, we want to decide if there
exist $r \geq 1$ and
two finite sequences of indices $i_1, i_2, \ldots , i_r$
and $j_1, j_2, \ldots, j_r$ such that
$x_{i_1} \cdots x_{i_r} = y_{j_1} \ldots y_{j_r}$.   In other words,
we only demand that the number of words on each side be the same.

This case is also efficiently decidable, even when $A$ and $B$ are possibly
infinite regular languages.

\begin{theorem}
Let $M_1$ (resp., $M_2$) be an NFA with $s_1$ states and $t_1$ transitions
(resp., $s_2$ states and $t_2$ transitions).
We can decide in polynomial time (in $s_1, s_2, t_1, t_2$) whether there exists
$k$ such that $L(M_1)^k \ \cap \ L(M_2)^k \not= \emptyset$.
\end{theorem}

\begin{proof}
First, we prove the (possibly surprising?) result that
$$ L = \bigcup_{k \geq 1} \left( L(M_1)^k \ \cap \ L(M_2)^k \right) $$
is a context-free language.

We construct a pushdown automaton $M$ accepting $L$.  On input $x$, our PDA
attempts to construct two same-length
factorizations of $x$:  one into elements of
$L(M_1)$, and one into elements of $L(M_2)$.  To ensure the factorizations
are really of the same length, we use the stack of the PDA
to maintain a counter that records the absolute value of the difference
between the number of factors in the first factorization and the number
of factors in the second.    The appropriate sign of the difference 
is maintained in the state of the PDA.

As we read $x$, we simulate the NFA's $M_1$ and $M_2$.  If we reach a final
state in either machine, then we have the option (nondeterministically)
to deem this the
end of a factor in the appropriate factorization, and update the stack
accordingly, or continue with the simulation.
We accept if the stack records a difference of $0$ --- that is,
if the stack contains no counters and only the initial stack symbol
$Z_0$ --- and we are in a final state in both machines (indicating that
the factorization is complete into elements of both $L_1$ and $L_2$).

Thus we have shown that $L$ is context-free.   Furthermore, our
PDA has $O(s_1 s_2)$ states and $O(t_1 t_2)$ transitions. It uses
only two distinct stack symbols --- the counter and the initial stack
symbol --- and never pushes more than one additional symbol on the stack
in any transition.
Such a PDA can
be converted to a context-free grammar $G$, using the standard ``triple
construction'' \cite[Thm.\ 5.4]{Hopcroft&Ullman:1979},
using $O(s_1^2 s_2^2)$ states and $O(s_1^2 s_2^2 t_1 t_2)$ transitions.
Now we can test the emptiness of the language generated
by a context-free grammar of size $t$ in
$O(t)$ time, by removing useless symbols and seeing if any productions
remain \cite[Thm.\ 4.2]{Hopcroft&Ullman:1979}.

We conclude that it is decidable in polynomial time 
whether there exists
$k$ such that $L(M_1)^k \ \cap \ L(M_2)^k \not= \emptyset$.
\endpf
\end{proof}

\begin{remark}
There exist simple examples where 
$ L = \bigcup_{k \geq 1} \left( L(M_1)^k \ \cap \ L(M_2)^k \right) $ is not
regular.  For example, take $L(M_1) = b^* a b^*$ and 
$L(M_2) = a^* b a^*$.  Then 
$L = \lbrace x  \in \lbrace a,b\rbrace^* \ : \ |x|_a = |x|_b \geq 1 \rbrace$,
the language of nonempty strings with the same number of $a$'s and
$b$'s.

Furthermore, if $M_1, M_2, M_3$ are all NFA's, then the analogous language
$$ L = \bigcup_{k \geq 1} \left( L(M_1)^k \ \cap \ L(M_2)^k \ \cap \ 
L(M_3)^k \right) $$
need not be context-free.  A counterexample is given by taking
$L(M_1) = \lbrace b,c \rbrace^* a \lbrace b,c \rbrace^*$,
$L(M_2) = \lbrace a,c \rbrace^* b \lbrace a,c \rbrace^*$,
and
$L(M_3) = \lbrace a,b \rbrace^* c \lbrace a,b \rbrace^*$.
Then 
$$L = \lbrace x  \in \lbrace a,b,c\rbrace^* \ : \ |x|_a = |x|_b = |x|_c \geq 1
\rbrace,$$
which is clearly not context-free.
\end{remark}

\begin{remark}
Mike Domaratzki (personal communication) observes that the decision
problem ``given $M_1$, $M_2$, does there exist $k \geq 1$ such that
$L(M_1)^k \ \cap \ L(M_2)^k \not= \emptyset$'' becomes undecidable if
$M_1$ and $M_2$ are pushdown automata, by reduction from the problem
``given CFG's $G_1, G_2$, is $L(G_1) \ \cap \ L(G_2) \not= \emptyset$''
\cite[Theorem 8.10]{Hopcroft&Ullman:1979}.  Given $G_1$ and $G_2$,
we can easily create PDA's accepting $L_1 := L(G_1)\#$ and 
$L_2 := L(G_2)\#$, where
$\#$ is a new symbol not in the alphabet of either $G_1$ or $G_2$.  
Then $L_1^k \ \cap \ L_2^k \not= \emptyset$ for some $k \geq 1$
if and only if $L(G_1) \ \cap L(G_2) \not= \emptyset$.  A
similar result holds for the linear context-free languages
\cite{Baker&Book:1974}.
\end{remark}

We now turn to the question of, given regular languages $A$ and $B$,
determining the shortest string in 
$L = \bigcup_{k \geq 1} \left( A^k \ \cap \ B^k \right)$,
given that it is nonempty.       Actually, we consider a more general
problem, where we intersect more than two languages.  We start
by proving a result about directed graphs.

\newcommand{\abs}[1]{\lvert#1\rvert}
\newcommand{\norm}[1]{\lVert#1\rVert}
\newcommand{\mbf}{\mathbb{F}}
\newcommand{\mbr}{\mathbb{R}}
\newcommand{\mbz}{\mathbb{Z}}
\newcommand{\mbc}{\mathbb{C}}
\newcommand{\mbq}{\mathbb{Q}}
\newcommand{\ol}{\overline}
\newcommand{\bbm}{\begin{bmatrix}}
\newcommand{\ebm}{\end{bmatrix}}

\newcommand{\bpm}{\begin{pmatrix}}
\newcommand{\epm}{\end{pmatrix}}

\newcommand{\om}{\Omega}

\begin{lemma}
Suppose $G = (V, E)$ is a directed graph with edge weights in $\mbz^d$,
where the components of the edge weights are all bounded in absolute
value by $K$. Let $\sigma(p)$ denote the weight of a path $p$, obtained by
summing the weights of all associated edges. If $G$
contains a cycle $C \colon u \rightarrow u$ such that $\sigma(C) = {\bf 0} = (0,0,\ldots,0)$,
then $G$ also contains a cycle $C' \colon u \rightarrow u$ with
$\sigma(C') = {\bf 0} = (0,0,\ldots,0)$
and length at most $|V|^{d+1} K^d d^{d/2} (|V|^2 + d)$.
\label{luke}
\end{lemma}

\begin{proof}
For each vertex $v$ in the cycle $C$, break $C$ at the first occurrence of $v$. This gives us 
\begin{align*}
C &= P_1 P_2 P_3 \cdots P_k
\end{align*}
such that $P_1 \colon v_1 \ar v_2, P_2 \colon v_2 \ar v_3, \ldots, P_k \colon v_k \ar v_{k+1}$ where $\{ v_1, \ldots, v_{k} \}$ is the set of vertices visited by $C$. The final vertex, $v_{k+1}$, is the same as $v_{1}$ because $C$ is a cycle. Notice that $k \leq |V|$ because each vertex appears at most once in the list $v_1, \ldots, v_k$.  

For each $P_i \colon v_i \ar v_{i+1}$, generate a new path $\hat{P_i} \colon v_i \ar v_{i+1}$ by removing all simple subcycles. The length of $\hat{P_i}$ is at most $|V|$; otherwise some vertex is repeated, so we have not removed all subcycles. Recombine the $\hat{P_i}$'s into a cycle $T = \hat{P_1} \cdots \hat{P_k}$ having length $|T| \leq |V|k \leq |V|^2$. In addition to $T$, we have a list of simple subcycles $B_1, \ldots, B_{\ell}$ that we removed while generating the $\hat{P_i}$'s. 

Consider the cycles we can construct using $T, B_1, B_2, \ldots, B_\ell$. For any $B_i$, we know $T$ visits the starting vertex of $B_i$ because $T$ visits all the vertices in $C$. Therefore we can splice $B_i$ into $T$ at its starting vertex. Since $B_i$ is a cycle, we can insert it into $T$ any positive number of times. We can also append $T$ to the whole cycle as many times as we like. These techniques allow us to construct a cycle with weight
\begin{align*}
t \sigma(T) + b_1 \sigma(B_1) + \cdots + b_\ell \sigma(B_\ell)
\end{align*}
where $t \geq 1$ and $b_1, \ldots, b_n \geq 0$ are all integers. 

Recall that $T, B_1, \ldots, B_\ell$ were constructed by decomposing $C$. Each edge from $C$ exists somewhere in $T, B_1, \ldots, B_\ell$, so we have
\begin{align*}
{\bf 0} = \sigma(C) &= \sigma(T) + \sigma(B_1) + \cdots + \sigma(B_\ell) .
\end{align*}
This shows that it is possible to write {\bf 0} as an integer linear combination of $\sigma(T), \sigma(B_1), \ldots, \sigma(B_{\ell})$. Unfortunately, for each nonzero $b_i$ we have at least one copy of $B_i$, with length at most $|V|$. Since all the $b_i$'s are nonzero and $\ell$ is unbounded, the corresponding cycle has unbounded length. 
If we hope to find a bounded cycle by this technique then we need to bound the number of nonzero $b_i$'s. 
Let us approach the problem with linear programming. Construct a matrix $A \in \mbr^{d \times \ell}$ where the $i$th column is given by 
$A^{(i)} = \sigma(B_i)$.
Let $b \in \mbr^{d}$ be the column vector $\sigma(T)$. We are looking for solutions to the problem
$$Ax = b , \ \ \ x \geq 0, \ \ \ x \in \mbr^{\ell} .$$
This is just the feasible set of a linear program in standard equality form. We saw earlier that it has the feasible solution $x = \begin{pmatrix} 1 & 1 & \cdots & 1 & 1 \end{pmatrix}^{T}$. Note that if $A$ is not full rank then we remove linearly dependent rows until we have a full rank matrix, and proceed with a matrix of rank $d' \leq d$.

Linear programming theory tells us a feasible problem of this form has a \emph{basic} feasible solution $x^{*}$ with at most $d$ nonzero entries. Without loss of generality (relabelling if necessary), take all but the first $d$ entries of $x^{*}$ to be zero. Letting $\hat{A}$ be the first $d$ columns of $A$, the basic solution $x^{*}$ satisfies the following equation:
\begin{align*}
\hat{A} \begin{pmatrix} x_1^{*} \\ \vdots \\ x_d^{*} \end{pmatrix} &= b; \\
\sigma(B_1) x_1^{*} + \cdots + \sigma(B_{d}) x_d^{*} &= -\sigma(T) .
\end{align*}

We are not done yet because the $x_{i}^{*}$s are real numbers and we need an integer linear combination. Cramer's rule gives an explicit solution for each coefficient, $x_i^{*} = \frac{\det(\hat{A}_{i})}{\det(\hat{A})} = \frac{|\det(\hat{A}_{i})|}{|\det(\hat{A})|}$, where $\hat{A}_{i}$ is the matrix $\hat{A}$ with the $i$th column replaced by $b$. Note that $\hat{A}$ and $\hat{A}_{i}$ are integer matrices, so their determinants are integers and $x_i^{*}$ is a rational number. When we multiply through by $|\det(\hat{A})|$, all the coefficients will be positive integers:
\begin{align*}
\sigma(B_1) |\det(\hat{A}_1)| + \cdots + \sigma(B_d) |\det(\hat{A}_{d})| + \sigma(T) |\det(\hat{A})| &= {\bf 0} .
\end{align*}

We can bound the determinants with Hadamard's inequality, which says that the determinant of a matrix $M$ is bounded by the product of the norms of its columns. Each $B_i$ is a simple cycle, so $|B_i| \leq |V|$. It follows that any entry of $\sigma(B_i)$ is at most $|V|K$, so $\|\sigma(B_i)\| \leq |V|K \sqrt{d}$. On the other hand, $T$ has length at most $|V|^2$, giving $\| \sigma(T) \| \leq |V|^2 K \sqrt{d}$. Combining these estimates gives 
$|\det(\hat{A}_{i})| \leq |V|^d K^d d^{d/2}$ for all $i$ and
$|\det(\hat{A})| \leq |V|^{d+1} K^d d^{d/2}$.
Now we construct the cycle $C'$ from this linear combination, with $|\det(\hat{A})|$ copies of $T$ and $|\det(\hat{A}_{i})|$ copies of each $B_i$. By construction, $C'$ has weight {\bf 0} and its length is bounded as follows:
\begin{align*}
|C'| &= |\det{A}| |T| + \sum_{i=1}^{d} |\det{A}_{i}| |B_i| \\
&\leq |V|^{d+1} K^d d^{d/2} |V|^2 + \sum_{i=1}^{d} |V|^d K^d d^{d/2} |V| \\
&= |V|^{d+1} K^d d^{d/2} (|V|^2 + d).
\end{align*}
\endpf
\end{proof}

\begin{corollary}
Consider a generalization of the third problem to $d$ languages 
$L_1, L_2, \ldots, L_d$ accepted
by NFA's having $s_1, \ldots, s_d$ states,
respectively. If 
$$\bigcup_{k \geq 1} (L_1^k \ \cap \cdots \ \cap \ L_d^k)$$
is nonempty, then the shortest string in
the language has length bounded by 
\begin{displaymath}
O( s^{d} (d-1)^{(d-1)/2} (s^{2} + d - 1) ),
\end{displaymath}
where $s := (s_1 + 1)(s_2 + 1) \ldots (s_n + 1)$. 
\end{corollary}

\begin{proof}
We discuss the case $d = 2$, and then briefly indicate how this is generalized
to the general case.

First we discuss an
automaton $M_K = (Q_K, \Sigma, \delta_K, q_{K}, F_K)$
accepting $K = A^* \ \cap \ B^*$ which is a slight variant of the
construction given in the proof of Theorem~\ref{ab}, above.

Suppose we are given a regular language $A$ (resp., $B$) accepted by
an NFA $M_1$ (resp., $M_2)$.  Without loss of generality, we will
assume that $M_1$ (resp., $M_2$) has no transitions into its initial
state.    This can be accomplished, if necessary,
by adding one new state with transitions
out the same as the transitions out of the initial state, and redirecting
any transitions into the initial state to the new state.  If the original
machine had $s$ states, then the new machine has at most
$s+1$ states.  Call these new machines
$M'_1 = (Q_1, \Sigma, \delta_1, q_1, F_1)$ and
$M'_2 = (Q_2, \Sigma, \delta_2, q_2, F_2)$.

Next we create an NFA-$\epsilon$ $M''_1 = (Q_1, \Sigma, \delta'_1, q_1, F'_1)$
by adding an $\epsilon$-transition
from every final state of $M'_1$ back to its initial state, and by
changing the set of final states
to be $F'_1 = \lbrace q_1 \rbrace$.  This new machine $M''_1$ accepts
$A^*$.    We carry out a similar construction on $M'_2$ obtaining $M''_2$
accepting $B^*$.

Finally, mimicking the construction of Theorem~\ref{ab} we create 
an NFA-$\epsilon$ $M_K$ accepting $K = A^* \ \cap \ B^*$ using the direct
product construction outlined above on $M''_1$ and $M''_2$.
Note that $M_K$ has at 
most $(s_1+1)(s_2 + 1)$ states and has exactly one
accepting state, which is its initial state.

We define the edge weights of $M_k$ to be $\Zee$ as follows.
An explicit $\epsilon$-transition in $M'_1$ or $M'_2$ marks the
end of a word, so each explicit $\epsilon$-transition taken
in $M'_1$ back to the start gets weight $+1$, while 
each explicit $\epsilon$-transition in $M'_2$ back to the start
gets weight $-1$.  In this way we keep track of the difference between
the number of factors used in $L(M'_1)$ and $L(M'_2)$.

For the general case, we form the intersection automaton as before, and
define the $i$'th coordinate of
$\sigma(P)$, for $1 \leq i < d$, to be the difference in the
number of $\epsilon$-transitions taken in $M'_1$ and $M'_{i+1}$.
Now just apply Lemma~\ref{luke} to get the desired bound.
\endpf
\end{proof}

When $d = 2$, we can improve on the result of the previous lemma:

\begin{theorem}
If $d = 2$, then the length of the cycle $C'$ in Lemma~\ref{luke} is at most $2K|V|^{2}$.
\end{theorem}

\begin{proof}
Remove simple cycles $B_1, B_2, \ldots, B_\ell$ from $C$ until are we left with $R$, which has no proper subcycles. It follows that $R$ must be a simple cycle, so we have decomposed $C$ into simple subcycles. Note that the weight of $C$ is the sum of the weights of all the $B_i$'s and $R$. 

If $R$ has weight 0 then take $C' = R$. We are done because $R$ has
length at most $|V| \leq 2K|V|^2$. If $R$ has nonzero weight then the positive and negative cases are identical so take $R$ to have positive weight without loss of generality. Then there
must be some $B_i$ with negative weight, otherwise the sum of the
weights of the $B_i$'s and $R$ would be positive, but $C$ has weight 0.
Call the negative weight cycle $S$.

If $R$ and $S$ have some vertex in common, then we can splice $\sigma(R)$ copies of $S$ into $-\sigma(S)$ copies of $R$ to get a cycle $C'$ of weight 0. Since $\sigma(R) \leq K|R|$ and $\sigma(S) \leq K|S|$, the cycle has length 
$|\sigma(R)| |S| + |\sigma(S)| |R| \leq 2K |R| |S| \leq 2K |V|^2$.

Otherwise, $R$ and $S$ have no vertex in common so we need to find some way to get from $R$ to $S$ and back again. Clearly $C$ passes through every vertex in $R$ and $S$, but we want a shorter cycle. Let $T$ be the shortest cycle that passes through some vertex in $R$ and some vertex in $S$. We will split $T$ into $\alpha$, the piece from $R$ to $S$, and $\beta$, the piece from $S$ to $R$. 

We know that $R, S$ are simple, and $\alpha, \beta$ must be simple or we could make a shorter cycle $T$ by making them shorter. Therefore, any vertex in $V$ occurs at most four times in $R$, $S$ and $T$, once for each of $R, S, \alpha, \beta$. But $R$ and $S$ have no vertices in common, so each vertex occurs at most three times in $R$, $S$ and $T$.

Now if some vertex $v$ occurs three times in $R$, $S$ and $T$, then it must be in $\alpha$, $\beta$ and either $R$ or $S$ (without loss of generality, let it be in $R$). Then we can remove a prefix of $\alpha$ up to $v$, producing $\hat{\alpha}$. Similarly, remove a suffix of $\beta$ starting from $v$, giving $\hat{\beta}$. Then $\hat{\alpha} \hat{\beta}$ is a shorter cycle that visits $v \in R$ and still visits $S$, contradicting the minimality of $T$. Therefore any vertex $v$ occurs at most twice in $R$, $S$ and $T$, so $|R| + |S| + |T| \leq 2|V|$. 

Let us combine $T$ with $R$ if $T$ has positive weight and $S$ if $T$ has negative weight to produce a cycle $Y$. Either $R$ or $S$ is left over, call it $X$. Note that $X$ and $Y$ have opposite sign weights, and also have a vertex in common. As before, we combine $|\sigma(X)|$ copies of $Y$ with $|\sigma(Y)|$ copies of $X$ to produce a cycle $C'$ of length at most $2 K |X| |Y|$. Under the constraint $|X| + |Y| = |R| + |S| + |T| \leq 2|V|$, the length $2 K |X| |Y|$ is maximized when $|X| = |Y| = |V|$, with maximum value $2 K |V|^2$, completing the proof. 
\endpf
\end{proof}

     Finally, we prove an improvement for the unary case.

\begin{proposition}
Let $A, B$ be nonempty finite languages over a unary alphabet, say
$A = \lbrace a^{m_1}, \ldots, a^{m_r} \rbrace$ and
$B = \lbrace a^{n_1}, \ldots, a^{n_s} \rbrace$.  Then
$A^k \ \cap \ B^k \not= \emptyset$ for some $k \geq 1$
iff $\min_{1 \leq i \leq r} m_i \leq \max_{1 \leq j \leq s} n_j$
and $\min_{1 \leq j \leq s} n_j \leq \max_{1 \leq i \leq r} m_i$.
If both conditions hold, then $A^k \ \cap \ B^k \not= \emptyset$ for some
$k < \max(m_1, \ldots, m_r, n_1, \ldots, n_s)$, and this bound is tight.
\end{proposition}

\begin{proof}
Suppose $\min_{1 \leq i \leq r} m_i > \max_{1 \leq j \leq s} n_j$.
Then every element of $A^k$ will be of length greater than every
element of $B^k$.  Similarly, if
$\min_{1 \leq j \leq s} n_j \leq \max_{1 \leq i \leq r} m_i$,
then every element of $B^k$ will be of length greater than every
element of $A^k$.  Hence if either condition holds, we have
$A^k \ \cap \ B^k = \emptyset$ for all $k \geq 1$.

Now suppose $\min_{1 \leq i \leq r} m_i \leq \max_{1 \leq j \leq s} n_j$
and $\min_{1 \leq j \leq s} n_j \leq \max_{1 \leq i \leq r} m_i$.
Then there exist $a^l, a^n \in A$ and $a^m \in B$ such that $l \leq m \leq n$.
Choose $i = n-m$ and $j = m-l$.  Then
$A^{i+j}$ contains
$(a^l)^i (a^n)^j = a^{li+nj} = a^{ln-lm+nm-nl} = a^{m(n-l)}$.
And $B^{i+j}$ contains $(a^m)^{i+j} = a^{m(n-l)}$.
So for $k = i+j$ we get $A^k \ \cap \ B^k \not= \emptyset$.
Now $i-j = n-l < n \leq \max(m_1, \ldots, m_r, n_1, \ldots, n_s)$.

The bound is tight, as can be seen by taking
$A = \lbrace a, a^n \rbrace$ and
$B = \lbrace a^{n-1} \rbrace$.  Then the least $k$ such that
$A^k \ \cap \ B^k \not= \emptyset$ is $k = n-1$.  
\endpf
\end{proof}

%
%

\end{document}